\newcommand{\arxivversion}[1]{#1}
\newcommand{\shortversion}[1]{}

\shortversion{
 \documentclass{aamas2013_extendedabstract}
 \newcommand{\myqed}{\qed}
}
\arxivversion{
 \documentclass[twocolumn]{article}
 \usepackage{authblk}
 \usepackage{amsthm}
 \newcommand{\myqed}{}
}
\usepackage{times}
\usepackage{graphicx}
\usepackage{latexsym}
\usepackage{amsmath}
\usepackage{amssymb}
\usepackage{mathrsfs}

\newtheorem{theorem}{Theorem}
\newtheorem{definition}{Definition}

\newtheorem{corollary}{Corollary}

\usepackage{xspace}

\input{epsf}

\begin{document}

\newcommand{\mydiv}{\mbox{\rm div}}
\newcommand{\mymod}{\mbox{\scriptsize \rm \ mod \ }}
\newcommand{\mymin}{\mbox{\rm min}}
\newcommand{\mymax}{\mbox{\rm max}}
\newcommand{\calFR}{{\cal FR}}
\newcommand{\calF}{{\cal F}}
\newcommand{\calC}{{\cal C}}
\newcommand{\calA}{{\cal A}}
\newcommand{\calL}{{\cal L}}
\newcommand{\dom}{{dom}}
\newcommand{\real}{\ensuremath{\mathbb{R}}}
\newcommand{\nat}{\ensuremath{\mathbb{N}}}
\newcommand{\ZZ}{\ensuremath{\mathbb{Z}}}
\newcommand{\upw}{UPW}
\newcommand{\unw}{UNW}
\newcommand{\wpw}{WPW}
\newcommand{\wnw}{WNW}

\newcommand{\mf}{{\mathscr F}}
\newcommand{\mc}{\mathcal C}
\newcommand{\ra}{\rightarrow}

\newcommand{\set}{\mathcal}
\newcommand{\myset}[1]{\ensuremath{\mathcal #1}}
\newcommand{\myOmit}[1]{}
\newcommand{\tighter}{\mbox{$\preceq$}}
\newcommand{\stighter}{\mbox{$\prec$}}
\newcommand{\incomparable}{\mbox{$\bowtie$}}
\newcommand{\equivalent}{\mbox{$\equiv$}}

\newcommand{\reg}{\mbox{$RE$}}
\newcommand{\mreg}{\mbox{$MR$}}
\newcommand{\hprs}{\mbox{$HPRS$}}
\newcommand{\reglo}{\mbox{$LO$}}
\newcommand{\ps}{\mbox{$PS$}}
\newcommand{\tc}{\mbox{$ST$}}
\newcommand{\among}{\mbox{$AD$}}
\newcommand{\lse}{\mbox{$LG$}}
\newcommand{\lser}{\mbox{$LG_R$}}
\newcommand{\cs}{\mbox{$CS$}}
\newcommand{\csdc}{\mbox{$CS_{DC}$}}
\newcommand{\fl}{\mbox{\sc $FB$}}
\newcommand{\flS}{\mbox{\sc $FB_S$}}
\newcommand{\amongS}{\mbox{$AD_S$}}

\newcommand{\gsc}{\mbox{\sc Gsc}}
\newcommand{\gcc}{\mbox{\sc Gcc}}
\newcommand{\GCC}{\mbox{\sc Gcc}}
\newcommand{\AllDifferent}{\mbox{\sc AllDifferent}}

\newcommand{\nina}[1]{{#1}}

\newcommand{\SLIDE}{\mbox{\sc Slide}}
\newcommand{\SLIDINGSUM}{\mbox{\sc SlidingSum}}
\newcommand{\REGULAR}{\mbox{\sc Regular}}
\newcommand{\TABLE}{\mbox{\sc Table}}
\newcommand{\CIRCREGULAR}{\mbox{$\mbox{\sc Regular}_{\odot}$}}
\newcommand{\WRAPREGULAR}{\mbox{$\mbox{\sc Regular}^{*}$}}
\newcommand{\STRETCH}{\mbox{\sc Stretch}}
\newcommand{\INCSEQ}{\mbox{\sc IncreasingSeq}}
\newcommand{\INC}{\mbox{\sc Increasing}}
\newcommand{\lseX}{\mbox{\sc Lex}}
\newcommand{\NFA}{\mbox{\sc NFA}}
\newcommand{\DFA}{\mbox{\sc DFA}}

\newcommand{\PRECEDENCE}{\mbox{\sc Precedence}}
\newcommand{\lseXVAR}{\mbox{\sc LexLeader}}
\newcommand{\lseXGENSET}{\mbox{\sc SetGenLexLeader}}
\newcommand{\lseXSETVAR}{\mbox{\sc SetLexLeader}}
\newcommand{\lseXMSETVAR}{\mbox{\sc MSetLexLeader}}
\newcommand{\lseXSETVAL}{\mbox{\sc SetValLexLeader}}
\newcommand{\lseXVAL}{\mbox{\sc ValLexLeader}}
\newcommand{\lseXVALVAR}{\mbox{\sc GenLexLeader}}
\newcommand{\VALVARLEX}{\mbox{\sc ValVarLexLeader}}
\newcommand{\NVALUES}{\mbox{\sc NValues}}
\newcommand{\USES}{\mbox{\sc Uses}}
\newcommand{\COMMONG}{\mbox{\sc Common}}
\newcommand{\CARDPATH}{\mbox{\sc CardPath}}
\newcommand{\RANGE}{\mbox{\sc Range}}
\newcommand{\ROOTS}{\mbox{\sc Roots}}
\newcommand{\AMONG}{\mbox{\sc Among}}
\newcommand{\ATMOST}{\mbox{\sc AtMost}}
\newcommand{\ATLEAST}{\mbox{\sc AtLeast}}
\newcommand{\ATMOSTSEQ}{\mbox{\sc AtMostSeq}}
\newcommand{\ATLEASTSEQ}{\mbox{\sc AtLeastSeq}}
\newcommand{\AMONGSEQ}{\mbox{\sc AmongSeq}}
\newcommand{\SEQUENCE}{\mbox{\sc Sequence}}
\newcommand{\GENSEQUENCE}{\mbox{\sc Gen-Sequence}}
\newcommand{\SEQ}{\mbox{\sc Seq}}
\newcommand{\myelement}{\mbox{\sc Element}}
\newcommand{\LEX} {\mbox{\sc Lex}}
\newcommand{\REPEAT} {\mbox{\sc Repeat}}
\newcommand{\REPEATONE} {\mbox{\sc RepeatOne}}
\newcommand{\STRETCHREPEAT} {\mbox{\sc StretchRepeat}}
\newcommand{\STRETCHONEREPEAT} {\mbox{\sc StretchOneRepeat}}
\newcommand{\STRETCHONEREPEATONE} {\mbox{\sc StretchOneRepeatOne}}
\newcommand{\SETSIGLEX} {\mbox{\sc SetSigLex}}
\newcommand{\SETPREC} {\mbox{\sc SetPrecedence}}

\newcommand{\SOFTATMOSTSEC} {\mbox{\sc SoftAtMostSequence}}
\newcommand{\ATMOSTSEC} {\mbox{\sc AtMostSequence}}
\newcommand{\SOFTATMOST} {\mbox{\sc SoftAtMost}}
\newcommand{\SOFTSEQ} {\mbox{\sc SoftSequence}}
\newcommand{\SOFTAMONG} {\mbox{\sc SoftAmong}}
\newcommand{\SEQCYC}{\mbox{$\mbox{\sc CyclicSequence}$}}

\newcommand{\ATMOSTSEQCYC}{\mbox{$\mbox{\sc AtMostSeq}_{\odot}$}}
\newcommand{\ignore}[1]{}

\newcommand{\mytrue}{\textsc{true}\xspace}
\newcommand{\myfalse}{\textsc{false}\xspace}

{\makeatletter
 \gdef\xxxmark{%
   \expandafter\ifx\csname @mpargs\endcsname\relax 
     \expandafter\ifx\csname @captype\endcsname\relax 
       \marginpar{xxx}
     \else
       xxx 
     \fi
   \else
     xxx 
   \fi}
 \gdef\xxx{\@ifnextchar[\xxx@lab\xxx@nolab}
 \long\gdef\xxx@lab[#1]#2{{\bf [\xxxmark #2 ---{\sc #1}]}}
 \long\gdef\xxx@nolab#1{{\bf [\xxxmark #1]}}
}

\newcommand{\ms}{\mathcal S}
\newcommand{\ma}{\mathcal A}
\newcommand{\mv}{\mathcal V}
\newcommand{\rev}{\text{rev}}
\newcommand{\others}{\text{\it Others}}

\newcommand{\vote}[3]{\mbox{$#1 \! \succ \! #2 \! \succ \! #3$}\xspace}
\newcommand{\vvote}[4]{\mbox{$#1 \! \succ \! #2 \! \succ \! #3 \! \succ \! #4$}\xspace}

\newcommand{\new}[1]{{#1}}
\newcommand{\myvec}[1]{\vec{#1}}

\renewcommand{\restriction}{\mathord{\upharpoonright}}
\newcommand{\votingRule}{\textsc{Maj}2\xspace}

\pagestyle{plain}

\title{Possible and Necessary Winner Problem in Social Polls}

\arxivversion{
  \author[1,2]{Serge Gaspers\thanks{Email: sergeg@cse.unsw.edu.au}}
  \author[3]{Victor Naroditskiy\thanks{Email: vn@ecs.soton.ac.uk}}
  \author[1,2]{Nina Narodytska\thanks{Email: nina.narodytska@nicta.com.au}}
  \author[1,2]{Toby Walsh\thanks{Email: toby.walsh@nicta.com.au}}
  \affil[1]{NICTA, Sydney, Australia}
  \affil[2]{The University of New South Wales, Sydney, Australia}
  \affil[3]{School of Electronics and Computer Science, University of Southampton, UK}
  \date{}
}

\shortversion{
  \numberofauthors{4}
  \author{
  \alignauthor
  Serge Gaspers\\
	\affaddr{UNSW and NICTA}\\
	\affaddr{Sydney, Australia}\\
	\email{sergeg@cse.unsw.edu.au}
  \alignauthor
  Victor Naroditskiy\\
	\affaddr{University of Southampton}\\
	\affaddr{Southampton, UK}\\
	\email{vn@ecs.soton.ac.uk}
  \alignauthor
  Nina Narodytska\\
	\affaddr{NICTA and UNSW}\\
	\affaddr{Sydney, Australia}\\
	\email{nina.narodytska@nicta.com.au}
  \and
  \alignauthor
  Toby Walsh\\
	\affaddr{NICTA and UNSW}\\
	\affaddr{Sydney, Australia}\\
	\email{toby.walsh@nicta.com.au}
  }
}

\maketitle

\begin{abstract}
Social networks are increasingly being used to
conduct polls. We introduce a simple model of
such social polling. We suppose agents vote sequentially, but
the order in which agents choose to vote is not necessarily
fixed. We also suppose that an agent's vote is influenced by
the votes of their friends who have already voted.
Despite its simplicity, this model provides useful
insights into a number of areas including
social polling, sequential voting, and
manipulation. We prove that the number of candidates
and the network structure affect the computational complexity
of computing which candidate necessarily or possibly
can win in such a social poll.
For social networks with bounded treewidth and
a bounded number of candidates, we provide polynomial algorithms
for both problems. In other cases, we prove that
computing which candidates necessarily or possibly
win are computationally intractable.
\end{abstract}

\shortversion{
  
  
}

\sloppy
\section{Introduction}

A fundamental issue with voting is that agents may vote strategically.
Results like those of Gibbard-Satterthwaite demonstrate that, under modest
assumptions, strategic voting is likely to be possible \cite{gs1,gs2}.
However,
such results do not tell us how to vote strategically. A large body
of work in computational social choice considers how we compute
such strategic votes \cite{fpaimag10,fhhcacm10}.
Typically such work starts from some
strong assumptions. For example, it is typically assumed that the
manipulators have complete information about the other votes.
The argument given for this assumption is that computing
a strategic vote will only be computationally harder with incomplete
information. In practice, of course, we often only have
partial or probabilistic information \cite{waaai2007,cwxaaai11}.
It is also typically assumed that manipulators will vote in any
way that achieves their ends. However, in practice, agents may
be concerned about peer pressure and may not want to deviate
too far from either their true vote or that of their
peers \cite{oeaamas12}. Bikhhardani et al.~\cite{BikhchandaniHW92} identified
several factors that limit strategic voting by an individual
agent such as sanctions on deviation, and conformity
of preferences. A third strong assumption
is either that all voting happens simultaneously or that
the manipulators get to vote after all the other agents.
Again, in practice, this is often not the case.

These issues all come to a head in {\em social polling}. This
is a context in which voting meets social networks. Startups like Quipol
and GoPollGo use social networks to track
public opinions. Such polls are often not anonymous.
We can see how our friends have voted and this may
influence how we vote. By their very nature,
such polls also happen over time. The order in which agents vote can
therefore be important. The structure of social networks
is also important. For example, a distinctive feature
of social networks is the small world property which
allows members of these communities to share
information in a highly efficient and low cost manner.
A rumor started in the Twitter network reaches about 90\% of the network
in just 8 rounds of communication \cite{DoerrFF12}. In a
similar way, one member of a social network can quickly create and
publicize a poll among a large group of agents starting from his
friends. The massive size of social networks, like Facebook,
Twitter and Google+, gives statistically significant polls.

To study social polling, we set up a general model
that captures several important features of voting
within a social network. First, our model
uses the structure of the social network. How
an agent votes depends on how their friends vote. Second,
our model supposes agents vote sequentially and the
order in which they vote is not under their control.
For example, when you vote may depend on when one
of your friends chooses to invite you to vote.
Third, our model supposes that agents are influenced
by their friends. In fact, an agent's vote is some function
of their true preferences and of
the preferences revealed by the votes of their friends that 
have already voted. We can obtain different
instances of our model by choosing different functions.

To study this model, we consider a particular
instance that captures some of the features of
a Doodle poll. More precisely, each agent has a set of $k$
preferred candidates and is indifferent about other candidates.
Among these $k$ preferred candidates, one candidate is her top choice.
If a particular candidate among her $k$ preferred candidates
has a majority amongst her friends that have already voted,
then she mimics
their choice. Otherwise, she votes for her top choice.
Note that
any computational lower bounds derived for this particular instance also hold
for the general model.

Even though this instance of the model is simple and lacks some of the
subtleties of social influence in practice, it nevertheless
provides some valuable insights. For example, we
prove that it is computationally hard
to determine if a given candidate has necessarily won
a social poll, irrespective of how the remaining
agents vote. We also show that this intractability holds even if
the social graph has a simple structure
like  a disjoint union of paths. Of course, in practice
social influence is much more complex and
subtle. In addition, social graphs often
have much a richer structure than simple paths.
Finally, agents in general do not know
precisely how all the other agents will vote. 
However, all these issues will
only increase the computational complexity of 
reasoning about a social poll.

We focus here on computing the possible and necessary winners
of the social poll. A candidate
is a possible winner if there exists a voting order
such that this candidate is a plurality winner over the
cast votes. Similarly, a candidate is a necessary winner if he
is a plurality winner over the cast votes for each voting order.
The possible and necessary winner problems are interesting
in their own right. In addition, they provide insight into
several related and interesting problems. For example,
they are related to the control problem in which the
chair chooses an order of participation for the agents
that favors a particular outcome. In particular,
the chair can control the result of the election in
this way if and only if their desired candidate is
a possible winner.

\myOmit{
In this work, we define possible and necessary winner problems under
sequential voting in social networks and investigate their computational complexity.
We identify some basic factors that influence the computational
complexity of these problems: the number of the candidates and the structure of the social
network. We show that in most cases, it is computationally difficult
to determine the possible and necessary winners. However, if the structure of the
underlying social network has bounded treewidth and the number of candidates
is bounded then the problem becomes polynomially solvable.
}

\section{Problem Statement}

We consider a scenario where each agent votes for exactly one candidate.
We are given a social network graph $G=(V,E)$ whose $n$ vertices are the agents $x_1,\dots,x_n$,
a set $\mc=\{c_1,\dots,c_m\}$ of $m$ candidates, a distinguished candidate $c^*\in \mc$,
and a choice function $h$, which for every agent $x_i$, every subset $S\subseteq N_G(x_i)$ of its
neighbors in $G$, and every vote of an agent in $S$, assigns the candidate that $x_i$ votes for.
Each agent casts exactly one vote according to the following model.
For a given \emph{voting order} $\pi = (x_{\pi(1)}, \dots, x_{\pi(n)})$,
let $S_i$ denote the set $\{ x_{j} : \pi^{-1}(j) < \pi^{-1}(i) \} \cap N_G(x_i)$,
i.e., the neighbors of $x_i$ that vote before $x_i$.
Each agent $x_i$ votes for the candidate that the choice function
$h$ assigns for the given candidate $x_i$, the subset $S_i$ and the votes of the agents in $S_i$.
The \emph{score} of a candidate $c$ is the number of agents that vote $c$ in the voting order $\pi$.
A candidate $c\in \mc$ is a 
(co-)winner in the voting order $\pi$ if no other candidate has 
higher score than $c$.
A candidate is a \emph{possible winner} if there exists a voting order where $c$ is a winner.
A candidate is a \emph{necessary winner} if for every voting order, $c$ is a winner.

\medskip
\noindent\textbf{Refined model.}
We introduce a particular instance of the choice function $h$.
This is defined via two \emph{preference functions} $p_1:V\rightarrow \mc$ and $P:V\rightarrow 2^{\mc}$.\
Each agent $x\in V$ has a set $P(x)\subseteq \mc$ of $k$ preferred candidates, where $k>1$ is a constant.
Among the preferred candidates, one candidate $p_1(x) \in P(x)$ is the top preferred candidate.
Let $x$ be an agent and $S$ be the subset of $N_G(x)$ that voted before $x$.
If there exists a candidate $c\in P(x)$ such that more than half of the agents
from $S$ voted for $c$, then $x$ votes for $c$. Otherwise, $x$ votes for $p_1(x)$.
%
Note that all complexity lower bounds for this refined model also hold in the general model.

The \emph{unweighted possible winner ($\upw$)} problem is
to determine for an instance as described above whether $c^*$ is a possible winner.
Similarly, the \emph{unweighted necessary winner ($\unw$)} problem is to
determine whether $c^*$ is a necessary winner.
The \emph{weighted possible/necessary winner ($\wpw/\wnw$)} problems are defined similarly,
except that integer weights are associated with agents and the score of a candidate is
the sum of the weights of the agents that voted him.

\section{Overview of Results}

We will show that
the computational complexity of the possible and necessary winner
problem depends on the structure of the underlying social graph
and the number of candidates.
In particular, we prove that if the underlying social graph
has bounded treewidth and the number of candidates is
bounded then the unweighted possible and necessary winner
problems can be solved in polynomial time (Corollary~\ref{cor:unweightedTW}).
The degree of the polynomial bounding the running time of this algorithm is a function
of the number of candidates and the treewidth of the social network graph.
We give evidence that this cannot be avoided by showing that the problem is not finite-state.
For arbitrary social network graphs and a bounded number of candidates,
the weighted possible winner problem is NP-complete (Theorem~\ref{wpw_bpwg_alt_constant}), while
the weighted necessary winner problem is polynomial (Corollary~\ref{cor:NecessaryWeightedTW}).
If we relax the restriction on the
treewidth, all problems become computationally intractable (Theorems~\ref{upw_bpg_alt_constant}--~\ref{unw_bpg_alt_constant}).
Finally, we investigate these problems under the assumptions that
the number of candidates is unbounded and
the social graph is a disjoint union of paths. We show
that the unweighted possible winner problem is hard
even if the length of each path is at most one (Theorem~\ref{upw_cng_alt_unbounded}). By contrast, the
necessary winner problem is polynomial (Corollary~\ref{cor:NecessaryWeightedTW})
under the assumption that the number of candidates is unbounded
and the underlying social graph
has bounded treewidth.
Our results also demonstrate that the possible winner
problem is inherently computationally harder than
the necessary winner problem. This is not surprising as
the necessary winner problem requires much stronger
conditions to be satisfied for a candidate to be a necessary winner.
Table~\ref{t:results} summarizes our results.

\begin{table*}
\center{
{
\begin{tabular}{|c|c|c|c|c|c|c|}
  \hline
   $\#$ cands  & graph class &  $\upw$  &    $\unw$ & $\wpw$& $\wnw$\\
   \hline
   \hline
  $O(1)$& bounded treewidth & P (Cor~\ref{cor:unweightedTW})  &  P (Cor~\ref{cor:unweightedTW})  &  NPC (Thm~\ref{wpw_bpwg_alt_constant}) &  P (Cor~\ref{cor:NecessaryWeightedTW})\\
   &  & &   & (paths of length $\le 2$) & \\
   $O(1)$& bipartite & NPC (Thm~\ref{upw_bpg_alt_constant})  &   co-NPC (Thm~\ref{unw_bpg_alt_constant}) &  NPC  (Thm~\ref{upw_bpg_alt_constant})  & co-NPC ( Thm~\ref{unw_bpg_alt_constant})\\
   \hline
    \hline
  $O(n)$& bounded treewidth  & NPC (Thm~\ref{upw_cng_alt_unbounded})  &  P (Cor~\ref{cor:NecessaryWeightedTW})  &  NPC
  (Thm~\ref{upw_cng_alt_unbounded})  & P (Cor~\ref{cor:NecessaryWeightedTW}) \\
  &  &  (paths of length $1$)  &    &  (paths of length $1$)  & \\

  \hline
\end{tabular}  \caption{Overview of results\label{t:results}}}
}
\end{table*}

\section{Related work}

The possible and necessary winner problems were introduced
in the context of simultaneous voting to capture
uncertainty 
in preferences.
For example, due to incomplete preference elicitation,
we may have only have partial orders over the candidates
as the preferences of the voters.
Konczak and Lang considered two questions over a profile
with partial orders~\cite{klijcai2005}. Let $c^*$ be a distinguished candidate.
The first question is whether there is an extension of the partial orders to linear orders
such that the candidate $c^*$ wins. The second question is whether the candidate $c^*$ wins
for every extension of the partial orders to linear orders.
Our definitions of possible and necessary winner problems are inspired by
these two questions, but with uncertainty introduced by the voting order.

Xia and Conitzer~\cite{xcaaai08} identified connections between
possible and necessary winner problems and a number of
important problems in computational social choice,
including manipulation and preference elicitation problems.
The computational complexity  of the possible and necessary winner problems
under many commonly used voting rules has been extensively
investigated~\cite{xcaaai08,waaai2007}. If the number of candidates
is bounded and votes are unweighted then these problems
can be solved in polynomial time for any voting rule
that itself is polynomial~\cite{waaai2007,csljac2m007,prvwijcai2007}.
If the number of candidates is unbounded and votes are weighted,
these problems become computationally hard~\cite{waaai2007,csljac2m007}. 
Xia and Conitzer also investigated the setting where the number of candidates is unbounded and votes are unweighted~\cite{xcaaai08}.
They showed that the computational complexity in this
case depends on the voting rule.
Their results also demonstrate that the possible winner problem is computationally
harder than the necessary winner problem for many rules, including a class of
positional scoring rules, Maximin and Bucklin voting rules. We observe a similar relation
between the computational complexity of  possible and necessary winner problems in
social polls.

Perhaps closest to this work is Alon {et al.}~\cite{alon12}. However, the problems studied
there are rather different. In their model, agents have private preferences and vote strategically.
An agent experiences disutility if the winning candidate differs from his vote.
The authors derive an equilibrium voting strategy as a function of previously cast votes.
As soon as a candidate accumulates a (small) lead, all future votes are cast in his favor
independent of private preferences. This ``herding" behavior is compared across simultaneous and sequential
voting equilibria. Simultaneous and sequential voting mechanisms have also been compared based
on how well preferences are aggregated in equilibria of corresponding games~\cite{Dekel00,Battaglini07}.
Preference aggregation over multiple issues in the presence of influence has also been studied
by Maudet {et al.}~\cite{MaudetPVR12}.

\section{Preliminaries}

%
%

\medskip
\noindent\textbf{Graph theory.}
%
%
We refer to \cite{Diestel10} for basic notions of graphs and digraphs.
The path on $k$ vertices is denoted $P_k$.
For our algorithmic results, a central notion is the treewidth
of graphs \cite{RobertsonS84}.
A \emph{tree decomposition} of a graph $G=(V,E)$ is a pair
$(\{B_i : i\in I\},T)$
where the sets $B_i \subseteq V$, $i\in I$, are called \emph{bags} and $T$ is a tree with elements
of $I$ as nodes
such that:
\begin{enumerate}
  \item for each edge $uv\in E$, there is an $i\in I$ such that $\{u,v\}
\subseteq B_i$, and
\item for each vertex $v\in V$, $T[\{i\in I: v\in B_i\}]$ is a 
tree with at least one node.%
\end{enumerate}
The \emph{width} of a tree decomposition is $\max_{i \in I} |B_i|-1$.
The \emph{treewidth} of $G$
is the minimum width taken over all tree decompositions
of $G$. 

\medskip
\noindent\textbf{Social network graph.}
Let $V$ be a set of voters.
A binary friendship relation $\calFR$ on $V$
is a collection of unordered pairs $\calF \subseteq V \times V$.
We consider a relation $\calFR$ that is symmetric, reflexive and complete.
Given  $\calFR$ we build the social network graph $G$ as follows.
For each agent $v_i \in V$ we introduce a vertex $v_{i}$.
We connect two vertices $v_{i}$ and $v_{j}$  iff $(v_i,v_j) \in \calFR$.

\medskip
\noindent\textbf{NP-complete problems.}
Our hardness reductions rely on the NP-completeness of several classic problems~\cite{GareyJ79}.
A {\sc  partition} instance contains a set of integers
$A = \{k_0,\ldots,k_{n-1}\}$ such that $\sum_{j=0}^{n-1}k_j=2K$.
The problem is to determine
whether there exists a
partition of these numbers into two sets which sum to $K$.
A {\sc 3-hitting set} instance contains two sets:
$Q=\{q_0,\ldots,q_{n-1}\}$ and $S=\{S_1,\ldots,S_t\}$, where $t\geq
2$ and for all $j\leq t$, $|S_j|=3$ and $S_j\subseteq Q$.
The problem is to determine whether there exists a
set $H$, a so-called \emph{hitting set}, of size at most $k$ such that
$H \cap S_i \neq \emptyset$, $i=1,\ldots,t$.
Consider a set of Boolean variables $X = \{x_1,\ldots,x_n\}$.
A \emph{literal} is either a Boolean variable $x_i$ or its negation $\bar{x}_i$.
A \emph{clause} is a disjunction of literals.
A Boolean formula in \emph{conjunctive normal form} (CNF) is
a conjunction of $m$ clauses, $\{c_1, \ldots,c_m\}$.
A {\sc $(3^{\leq},3^{\leq})$-SAT} instance is a CNF formula such that every
clause has at most 3 literals and each variable occurs at most 3 times.
The problem is to check whether there exists an instantiation of Boolean
variables $X$ to make the $(3^{\leq},3^{\leq})$-SAT  instance evaluate to
{\sc true}, which is an NP-complete problem~\cite{Tovey84}.

\section{Tractable cases}

In this section we describe algorithms for the polynomial time solvable cases
in Table \ref{t:results}.
To simplify the description, we use the concept of nice tree decompositions.
A tree decomposition $(\{B_i : i\in I\},T)$ is \emph{nice} if each node $i$ of $T$ is of one of four types:
\begin{description} \itemsep=0pt
  \item[Leaf node:] $i$ is a leaf in $T$ and $|B_i|=1$;
  \item[Insert node:] $i$ has one child $j$, $|B_i|=|B_j|+1$, and $B_j \subset B_i$;
  \item[Forget node:] $i$ has one child $j$, $|B_i|=|B_j|-1$, and $B_i \subset B_j$;
  \item[Join node:] $i$ has two children $j$ and $k$ and $B_i=B_j=B_k$.
\end{description}
\noindent
An algorithm by Kloks \cite{Kloks94} converts any tree decomposition into a nice tree decomposition of the same width in linear time. 

A \emph{score function} of $\mc$ is a function $\#:\mc\rightarrow \mathbb{N}$.
A score function $\#$ can be \emph{achieved} by an instance if there is a voting order where $c$ is voted by $\#(c)$ agents, for every candidate $c\in \mc$.

\begin{theorem}\label{thm:unweightedTW}
There is a polynomial time algorithm, which, given a social network graph $G=(V,E)$ with treewidth $t=O(1)$,
a set $\mc$ of $m=O(1)$ candidates, and preference functions $P$ and $p_1$,
computes all possible score functions that can be achieved by this instance.
\end{theorem}
\begin{proof}
By Bodlaender's algorithm \cite{Bodlaender96}, compute a minimum width tree decomposition of $G$ in linear time. Let $t$ denote the width of this tree decomposition. Using Kloks' algorithm \cite{Kloks94}, convert it into a nice tree decomposition of width $t$ with $O(n)$ nodes in linear time. Select an arbitrary leaf of this tree decomposition, add a neighboring empty bag $r$ and root the tree decomposition at $r$. Denote the resulting tree decomposition by $(\{B_i : i\in I\},T)$.

In the description of our algorithm, we denote by $G_{\downarrow i}$ the subgraph induced by the subset of all vertices occurring in $B_i$ and bags associated to descendants of $i$ in $T$.

First, observe that the vote of a given agent does not depend on the ordering of the agents that voted before her, but solely on which subset of her friends were ordered before her. Therefore, instead of storing partial orderings of agents that have already been processed, we may merely store acyclic orientations of subgraphs of the friendship graph, where an edge oriented from $x$ to $y$ represents that $x$ votes before $y$. Any linear ordering extending a given acyclic orientation of the friendship graph will produce the same voting outcome.

Our dynamic programming algorithm will process bottom-up from the leafs to the root of the tree decomposition.
The computation at an internal node $i$ looks up the already computed results stored at its children.
Note that we cannot afford to remember all oriented paths in all relevant orientations of $G_{\downarrow i}$ that were computed at descendants of node $i$. All we need to remember at node $i$ is whether for two vertices $x,y\in B_i$, our
computations rely on orientations of subgraphs of $G_{\downarrow i}$ that contain a directed path from $x$ to $y$.
If so, we remember that there is a path from $x$ to $y$ by adding an arc $(x,y)$ to a directed acyclic graph (DAG)
with vertex set $B_i$ to the local information stored at this node.
Additionally, for every edge $xy$ in $G[B_i]$, we also need to decide (resp., go over all possible decisions), whether $x$ votes before $y$, or $y$ votes before $x$. This is again stored by orienting the edge $xy$ accordingly.
Therefore, at a node $i$, we process all DAGs on the vertex set $B_i$ whose underlying undirected graphs are supergraphs of $G[B_i]$.
For such a DAG $D$, we also process all votes of the vertices in $B_i$ (a voting function $v: B_i \rightarrow \mc$), all potential scores of candidates resulting from the votes of vertices in $G_{\downarrow i}$ (a score function $\#: \mc \rightarrow \{0, \dots, n\}$).
In addition, in order to do a sanity check to determine whether an agent $x\in B_i$ has indeed cast her vote according to our model after we have seen the votes of all her friends, we store for each candidate in $P(x)\setminus p_1(x)$ how many friends voted that candidate (an \emph{influence function} $s$ mapping an agent $x\in B_i$ and a candidate $c\in P(x)\setminus p_1(x)$ to a natural number in $\{0,\dots,n\}$) and how many of her friends voted before her (an \emph{anterior function} $a: B_i \rightarrow \{0,\dots,n\}$).

A voting function $v: X \rightarrow \mc$ on a subset of agents $X\subseteq V$ is \emph{legal} if $v(x)\in P(x)$, for every agent $x\in X$.
A voting function $v:X\rightarrow \mc$ \emph{extends} a voting function $v':X'\rightarrow \mc$ if $X'\subseteq X$ and $v(x)=v'(x)$ for every $x\in X'$.
An anterior function $a:X\rightarrow \{0,\dots,n\}$ is compatible with an influence function $s: X\times \mc \rightarrow \{0, \dots, n\}$ if for every $x\in X$, we have that $\sum_{c \in P(x)\setminus p_1(x)} s(x,c) \le a(x)$.
A voting function $v$ is \emph{compatible} with two compatible anterior and influence functions $a:X\rightarrow \{0,\dots,n\}$ and $s: X\times \mc \rightarrow \{0, \dots, n\}$ if for every vertex $x\in X$ with $N_G(x) \subseteq X$, we have that $v(x)=c$ if there exists a $c\in P(x)\setminus p_1(x)$ such that $s(x,c)>a(x)/2$, and $v(x)=p_1(x)$ otherwise.
A voting function $v:X\rightarrow \mc$ is \emph{compatible} with a score function $\#:\mc\rightarrow \{0,\dots,n\}$ if for every candidate $c\in \mc$, $|\{x\in X: v(x)=c\}|=\#(c)$.
The function $s$ is \emph{compatible} with a DAG $D$ with vertex set $X$ and a voting function $v$ if for every agent $x\in X$ and every candidate $c\in P(x)\setminus p_1(x)$, we have that $s(x,c)            = |\{y\in N^-_D(x) : v(y) =  c\}|$.
The function $a$ is \emph{compatible} with       $D$                     if for every agent $x\in X$, $a(x) = |N^-_D(x)|$.
We say that $v$, $D$, $\#$, $s$, $a$ are \emph{mutually compatible} if $a$ is compatible with $s$, $v$ is compatible with $a$ and $s$, $v$ is compatible with $\#$, $s$ is compatible with $D$ and $v$, and $a$ is compatible with $D$.

The algorithm computes a table entry for every relevant set of parameters $(i,v,D,\#,s,a)$, which is a Boolean and is \mytrue\ if and only if there is an acyclic orientation $D_{\downarrow i}$ of $G_{\downarrow i}$ such that:
\begin{itemize}
\item if there are two vertices $x,y$ in $B_i$ and a directed path from $x$ to $y$ in $D_{\downarrow i}$, then the arc $(x,y)$ is in $D$,
\item the voting function $v: B_i\rightarrow \mc$ can be extended to a legal voting function $v':V(G_{\downarrow i}) \rightarrow \mc$, and
\item $v'$, $D_{\downarrow i}$, $\#$, $s$, $a$ are mutually compatible.
\end{itemize}

Now that we have identified the relevant information stored at each node of the tree decomposition, the actual dynamic programming recurrences are fairly straightforward. We only need to ensure that the computations rely on already-computed table entries that are compatible with the entry that is being computed. For simplicity, we disregard issues arising from out-of-bounds table parameters and undefined values by assuming those entries to be \myfalse.

\medskip
\noindent\textbf{Leaf.} Suppose $i$ is a leaf with $B_i = \{x\}$. We set $T(i,v,D,\#,s,a)$ to \mytrue\ if $D=(\{x\},\emptyset)$, $v:\{x\}\rightarrow \mc$ is legal, and $v$, $D$, $\#$, $s$, $a$ are mutually compatible, and to \myfalse\ otherwise.

\medskip
\noindent\textbf{Insert node.} Suppose $i$ is an insert node in $T$ with child $j$. Let $x$ be the unique agent in $B_i \setminus B_j$. We set $T(i,v,D,\#,s,a)$ to \myfalse\ if $v$ is not legal or $s(x,c)$ is not the number of $y\in N^-_D(x)$ such that $v(y)=c$, for every $c\in P(x)\setminus p_1(x)$, or $a(x) \ne |N^-_D(x)|$. Otherwise, set $T(i,v,D,\#,s,a) := T(j,v',D',\#',s',a')$ where:
\begin{itemize}
 \item $v' = v\restriction_{B_j}$,
 \item $D' = D-x$,
 \item $\#'$ is obtained from $\#$ by decrementing $\#(v(x))$ by one,
 \item $s'$ is obtained from $s\restriction_{B_j \times \mc}$ by decrementing $s(y,v(x))$ by one for every $y\in N^+_D(x)$ such that $v(x)\in P(y)\setminus p_1(y)$, and
 \item $a'$ is obtained from $a\restriction_{B_j}$ by decrementing $a(y)$ by one for every $y\in N^+_D(x)$.
\end{itemize}
Here, $f\restriction_{A}$ denotes the restriction of a function $f:B\rightarrow C$ to a subdomain $A\subseteq B$.

\medskip
\noindent\textbf{Forget node.} Suppose $i$ is a forget node in $T$ with child $j$. Let $x$ be the unique agent in $B_j \setminus B_i$. Since $x$ occurs only in $B_j$ and its descendants in $T$, all neighbors of $x$ are in $V(G_{\downarrow i})$. Therefore, we now do a sanity check and disregard all situations where $x$ does not vote according to our model. We set $T(i,v,D,\#,s,a)$ to \myfalse\ if $v$ is not legal, or $v(x)=p_1(x)$ but there exists a candidate $c\in P(x)\setminus p_1(x)$ with $s(x,c)>a(x)/2$, or $v(x)\ne p_1(x)$ but $s(x,c)\le a(x)/2$ for every candidate $c\in P(x)\setminus p_1(x)$. Otherwise it is obtained by computing a disjunction of all $T(j,v',D',\#',s',a')$ such that:
\begin{itemize}
 \item $v'$ extends $v$,
 \item $D=D'-x$,
 \item $\#=\#'$,
 \item $s=s'$,
 \item $a=a'$,
 \item if $v'(x)=p_1(x)$ then $s(x,c)\le a(x)/2$ for every $c\in P(x)\setminus p_1(x)$, and
 \item if $v'(x)\ne p_1(x)$ then $s(x,v'(x)) > a(x)/2$.
\end{itemize}

\medskip
\noindent\textbf{Join node.} Suppose $i$ is a join node in $T$ with children $j$ and $j'$. Since all agents that occur in both $G_{\downarrow j}$ and $G_{\downarrow j'}$, also occur in $B_i$, we can easily correct any overcounting resulting from summing values for the subproblems at $j$ and $j'$ when computing the functions $\#$, $s$, and $a$ at node $i$.
We set $T(i,v,D,\#,s,a)$ to be a disjunction over all
$T(j,v',D',\#',s',a') \wedge T(j',v'',D'',\#'',s'',a'')$ with:
\begin{itemize}
 \item $v=v'=v''$,
 \item $D=D'=D''$,
 \item $\#(c)=\#'(c)+\#''(c)-|\{x\in B_i : v(x)=c\}|$ for each $c\in \mc$,
 \item $s(x,c) = s'(x,c) + s''(x,c) - |\{ y\in N^-_D(x) : v(y) =  c \}|$ for each $x\in B_i$ and $c\in P(x)\setminus p_1(x)$, and
 \item $a(x) = a'(x) + a''(x) - |N^-_D(x)|$ for each $x\in B_i$.
\end{itemize}

After all table entries have been computed, we inspect the entries at the root node $r$ of $T$. Since $B_r$ is empty, all table entries associated with node $r$ have an empty voting function $v$, a vertex-less DAG $D$, and empty anterior and influence functions $a$ and $s$. The only relevant information still contained in these entries are the score functions $\#$ that can be achieved by the instance. The algorithm returns these score functions.

Let us now upper bound the number of table entries.
The number of nodes of $T$ is $O(n)$.
For each node $i$ of $T$, $|B_i|\le t$.
Thus, the number of legal voting functions $v: B_i \rightarrow \mc$ is at most $k^t$.
Denoting by $q_t$ the number of labeled directed acyclic graphs on $t$ nodes, $q_t$ can be expressed by the recurrence relation
\begin{align*}
 q_t = \sum_{k=1}^t (-1)^{k-1} \binom{t}{k} 2^{k (t-k)} q_{t-k}
\end{align*}
with $q_1=1$ \cite{HararyP73,Robinson73}.
Asymptotically, $q_t = O(t! 2^{\binom{t}{2}} 1.488^{-t})$ (see, e.g., \cite{Liskovets08}).
The number of distinct score functions is bounded by $n^{|\mc|}$.
The number of influence functions is bounded by $n^{t(k-1)}$.
The number of anterior functions is bounded by $n^t$.
Finally, the number of table entries is $O(n \cdot k^t \cdot t! 2^{\binom{t}{2}} \cdot n^{|\mc|} \cdot n^{t(k-1)} \cdot n^t)$.

Each table entry can be computed in time $O(n^{|\mc| + t k})$.
Indeed, the computations at the leaf and the insert nodes can be done in time $O(1)$.
A table entry computed at a forget node $i$ ranges over all legal extensions $v'$ of $v$ and all digraphs $D'$ such that $D=D'-x$.
Since $|V(D')| \le t$, there are $O(3^t)$ such digraphs: each vertex from $D$ is either not a neighbor or an in-neighbor or an out-neighbor of $x$ in $D'$.
The number of legal extensions of $v$ to the domain $B_i \cup \{x\}$ is $k$.
Thus, table entries at a forget node can be computed in time $O(3^t)$ which is in $O(n^{|\mc| + t k})$ if $n>1$.
Computations at join nodes range over all possibilities to sum $\#'(c)$ and $\#''(c)$ to $\#(c)+|\{x\in B_i : v(x)=c\}|$ for each $c\in \mc$,
all possibilities to sum $s'(x,c)$ and $s''(x,c)$ to $s(x,c)+|\{ y\in N^-_D(x) : v(y) =  c \}|$ for each $x\in B_i$ and each $c\in P(x)\setminus p_1(x)$, and all possibilities to sum $a'(x)$ and $a''(x)$ to $a(x)+|N^-_D(x)|$ for each $x\in B_i$. Thus, the computation of a table entry at a join node looks up $O(n^{|\mc|+t k})$ table values.
All in all, our algorithm has running time $O(n^{1+2|\mc|+2t k} \cdot k^t \cdot t! \cdot 2^{\binom{t}{2}}) = O(n^{1+2|\mc|+2t k} \cdot 2^{t\log k + t\log t + t^2})$.
\myqed \end{proof}

\noindent
After executing this algorithm, one can easily identify whether a candidate $c$ is a possible or necessary
winner by inspecting the score functions that can be achieved by the instance.

\begin{corollary}\label{cor:unweightedTW}
For any class of instances where the treewidth of the social network and the number of candidates are bounded by a fixed constant, the unweighted possible and necessary winner problems can be solved in polynomial time.
\end{corollary}

\noindent
Theorem~\ref{wpw_bpwg_alt_constant} shows that the weighted version of
the possible winner problem is NP-hard under the same restrictions.
The necessary winner problem can be reformulated as $m-1$ subproblems of the following type:
is there a voting order where candidate $d$ achieves a higher score than candidate $c$?
If some other candidate can achieve a higher score than our distinguished candidate $c^*$, then $c^*$ is not a necessary winner.
Testing whether a candidate $d$ can achieve a higher score than a candidate $c$ can be done by a slight variation of our previous algorithm, even for the weighted version of the problem and for an unbounded number of candidates.

\begin{corollary}\label{cor:NecessaryWeightedTW}
 The weighted necessary winner problem can be solved in polynomial time for social network graphs with treewidth $O(1)$.
\end{corollary}
\begin{proof}
 We need a polynomial time test of whether a candidate $d$ achieves a higher score than a candidate $c$.
 We modify the algorithm in the proof of Theorem~\ref{thm:unweightedTW} as follows.
 Remove the function $\#$ from the table parameters. Instead, each table entry is an integer, representing the maximum possible value of
 the score of candidate $d$ minus the score of candidate $c$ in this subinstance. This change implies some other changes in the computation of the table entries (a disjunction of table entries becomes a maximum, setting a table entry to \myfalse\ becomes setting its value to $- \infty$, etc.), all of which are straightforward.
 In the end, there is a voting order where $d$ achieves a higher score than $c$ if the unique table entry at the root of the tree decomposition is positive.
 Since all factors of the form $n^{|\mc|}$ in the running time bound of Theorem~\ref{thm:unweightedTW} are due to the table parameter $\#$, this variant is polynomial even for an unbounded number of candidates.
\myqed \end{proof}

Although the algorithm from Theorem \ref{thm:unweightedTW} is polynomial whenever $|\mc|$ and $t$ are upper bounded by a fixed constant, its running time seems prohibitive even for relatively small values of $|\mc|$ and $t$. This is largely due to the degree of the polynomial bounding the running time depending on $|\mc|$ and $t$. Therefore, a natural question is whether the problems can be solved in time $f(|\mc|,t) \cdot n^c$, where $c$ is a constant independent of $|\mc|$ and $t$, and $f$ is a function independent of $n$.
Formulated in the terms of \emph{multivariate complexity} \cite{DowneyF99,FellowsGR11,FlumG06,Niedermeier06}: are the problems fixed-parameter tractable (FPT) parameterized by $|\mc|+t$? We conjecture that they are $W[1]$-hard, and give supporting evidence in terms of finite-state properties of graphs \cite{vanBevernFGR12,BodlaenderFW92,FellowsL89}.

\begin{definition}
 An $l$-\emph{boundaried graph} is a triple $(V,E,B)$ with $(V,E)$ a
 simple graph, and $B \subseteq V$ an ordered subset of $l\ge 0$ vertices.
 Vertices in $B$ are called \emph{boundary vertices}.
\end{definition}

\begin{definition}
 The operation $\oplus$ maps two $l$-boundaried graphs $G$ and $H$, $l\ge 0$, to a graph $G \oplus H$, by taking the disjoint union of $G$ and $H$, then identifying corresponding boundary vertices, i.e., for $i = 1..l$, identifying the $i$th boundary vertex of $G$ with the $i$th boundary vertex of $H$, and removing multiple edges.
\end{definition}

If $F$ is an arbitrary family of (ordinary) graphs, we define the following canonical equivalence relation $\sim_{F,l}$ induced by $F$ on the set of $l$-boundaried graphs.

\begin{definition}
 $G_1 \sim_{F,l} G_2$ if and only if for all $l$-boundaried graphs $H$, $G_1 \oplus H \in F \Leftrightarrow G_2 \oplus H \in F$.
\end{definition}
The graph family $F$ is of finite index if $\sim_{F,l}$ has a finite number of equivalence classes for all $l\ge 0$.

Slightly abusing notation, we use the previously defined terms for instances of our problems instead of graphs.

\begin{theorem}\label{thm:notfinitestate}
 The class of unweighted instances where the social network graph has treewidth at most $1$, the number of candidates is at most $2$, and $c^*$ is a possible (respectively, necessary) winner is not of finite index.
\end{theorem}
\begin{proof}
 Let $F$ be this class of instances.
 We consider the equivalence relation $\sim_{F,0}$ and show that it has an infinite number of equivalence classes.
 For every positive integer $i$, define the $0$-boundaried instance $L_i$ whose social network graph is the path $P_i$ and every voter $x$ on this path has $P(x)=\{c^*,a\}$ and $p_1(x)=c^*$.
 For every positive integer $i$, define the $0$-boundaried instance $R_i$ whose social network graph is the path $P_i$ and every voter $x$ on this path has $P(x)=\{c^*,a\}$ and $p_1(x)=a$.
 If $i>j$, then $L_i \not\sim_{F,0} L_j$ since $c^*$ is a winner in $L_i \oplus R_i$ for every ordering of the voters, but $c^*$ is not a winner in $L_j \oplus R_i$ for any ordering of the voters.
 Thus, every $L_i$, $i\ge 0$, is in a different equivalence class of the relation $\sim_{F,0}$.
\end{proof}

Consequently, finite-state automata are not
amendable to give an FPT algorithm, even for the parameter treewidth when the number of candidates
is upper bounded by a constant. Intuitively, Theorem \ref{thm:notfinitestate} implies that the amount of information
that the usual kind of algorithms need to transmit when transitioning from one bag of the tree decomposition to the next
cannot be upper bounded by a function depending only on the width of the tree decomposition.
It could still be upper bounded by an FPT function though,
in which case the other standard algorithmic technique for bounded-treewidth instances, dynamic-programming, could still give an FPT algorithm. However, the following
theorem shows that the index cannot be upper bounded by an FPT function either.

\begin{theorem}\label{thm:notdp}
 For every integer $n$, the class of unweighted instances whose social network graph has $n$ vertices and treewidth at most $1$, the number of candidates is $k$, and $c^*$ is a possible (respectively, necessary) winner has index at least
 $\lfloor n/k \rfloor^{k-1}$.
\end{theorem}
\begin{proof}
 Let $F_n$ be this class of instances.
 We consider the equivalence relation $\sim_{F_n,0}$ and show that it has at least
 $\lfloor n/k \rfloor^{k-1}$ equivalence classes. Let $\ell:=\lfloor n/k \rfloor$.
 For positive integers $i_1, \dots, i_{k-1} \le \ell$, define the $0$-boundaried instance $L_{i_1,\dots,i_{k-1}}$ whose social network graph is a disjoint union of paths $P_{i_j}$, $j=1,\dots,k-1$, and every voter $x$ on the path $P_{i_j}$ has $P(x)=\{c^*,a_j\}$ and $p_1(x)=a_j$.
 For positive integers $i_1, \dots, i_{k}\le \ell$, define the $0$-boundaried instance $R_{i_1,\dots,i_{k}}$ whose social network graph is a disjoint union of paths $P_{i_j}$, $j=1,\dots,k$, and every voter $x$ on the path $P_{i_j}$ with $i_j<k$ has $P(x)=\{c^*,a_j\}$ and $p_1(x)=a_j$ and every voter $x$ on the path $P_{k}$ has $P(x)=\{c^*,a_1\}$ and $p_1(x)=c^*$.
 Now, if $(i_1,\dots,i_{k-1}) \neq (i'_1,\dots,i'_{k-1})$, then $L_{i_1,\dots,i_{k-1}} \not\sim_{F_n,0} L_{i'_1,\dots,i'_{k-1}}$. To see this, suppose, w.l.o.g., that $i_1<i'_1$. Then $c^*$ is a winner in $L_{i_1,\dots,i_{k-1}} \oplus R_{\ell -i_1,\dots,\ell-i_{k-1},\ell}$ for every ordering of the voters, but $c^*$ is not a winner in $L_{i'_1,\dots,i'_{k-1}} \oplus R_{\ell -i_1,\dots,\ell-i_{k-1},\ell}$ for any ordering of the voters.
 Thus, every $L_{i_1,\dots,i_{k-1}}$, $0\le i_j \le \ell$, is in a different equivalence class of the relation $\sim_{F_n,0}$.
\end{proof}
 
Thus, we have little hope that the running time of the algorithm from Theorem \ref{thm:unweightedTW}
can be improved significantly.

\section{Intractable cases}

We observe that an isolated agent that has no friends
always votes for her top preferred candidate.
To simplify notations, we call the score of a candidate that comes
from all isolated agents the \emph{basic} score.
Our intractability results hold even if each voter has two preferred
candidates. We denote the two preferred candidates
of a voter $(x,y)$, where $x$ is the top preferred candidate.

\begin{theorem}~\label{wpw_bpwg_alt_constant}
The weighted possible winner problem is NP-complete even
if the social network graph is a disjoint union of paths of length at most two,
the number of candidates is constant, and each agent has two preferred candidates.
\end{theorem}
\begin{proof}
\begin{figure}
\centering
\includegraphics[width=0.40\textwidth]{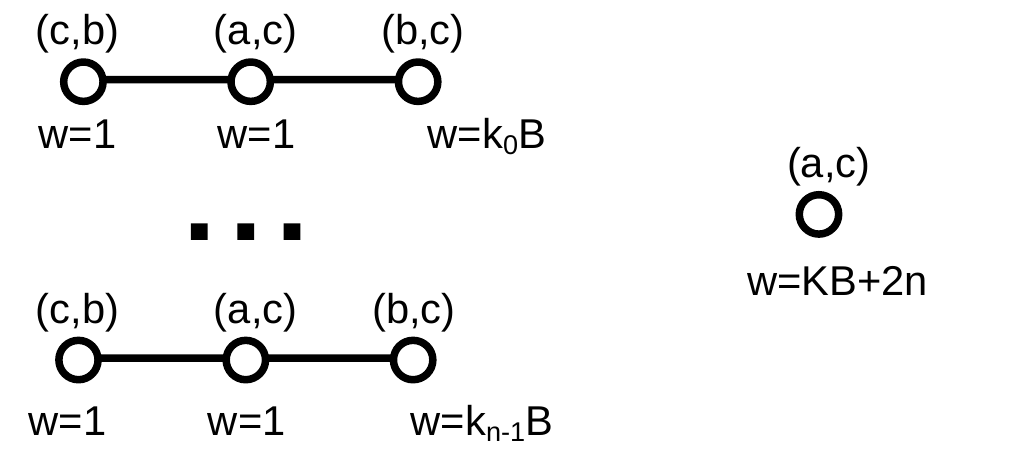}
\caption{\label{fig:proof_wpw_bpwg_alt_constant} The construction from Theorem~\ref{wpw_bpwg_alt_constant}}
\end{figure}

We reduce from an instance of the {\sc partition} problem
to $\wpw$ with three candidates $\{a,b,c\}$.

For each integer $k_j$, $j=0,\ldots,n-1$ we introduce 3 agents $3j,3j+1$
and $3j+2$, with preferences $(c,b), (a,c), (b,c)$, respectively.
The weights of the $(3j)$th agent and the $(3j+1)$th agent are one.
The weight of the $(3j+2)$th agent is $k_{j}B$, where $B$ is a large integer,
for instance $2n+1$.
Agents $3j,3j+1$
and $3j+2$ form the $j$th path of friends, $((3j,3j+1),(3j+1,3j+2))$,
that corresponds to the $k_{j}$th element.
We introduce an additional agent without friends, with
preferences $(a,c)$ and weight $KB+2n$. We ask whether $a$ is
a possible winner.
Figure~\ref{fig:proof_wpw_bpwg_alt_constant} illustrates the construction.

The basic score of $a$ is $KB+2n$.
The idea of the construction is
to make sure that the preferred candidate $a$
wins iff the weighted votes of $(3j+2)$th agents, $j=0,\ldots,n-1$,
are partitioned equally between candidates
$b$ and $c$.
Consider the $j$th path $((3j,3j+1), (3j+1,3j+2))$.
The $(3j+2)$th agent either votes for $b$ or for $c$
depending on the relative order of the candidates in this path.
As the weight
of the $(3j+2)$th agent is $k_{j}B$, either
$c$ or $b$ increases its score by $k_{j}B$.
Let $J$ be a set of paths such that the $(3j+2)$th agent
selects $b$, $j\in J$, and $J^c = \{0,\ldots,n-1\} \setminus J $
contains all paths such that the $(3j+2)$th agent selects $c$, $j\in J^c$.
Then the total weight that the candidate $b$ gets is
$\sum_{j \in J} k_{j}B = B\sum_{j \in J} k_{j}$.
If $\sum_{j \in J} k_{j} > K$ then the score of
$b$ is strictly greater than the maximum score of $a$.
Similarly, the total weight that the candidate $c$ gets
is $\sum_{j \in J^c} k_{j}B = B\sum_{j \in J^c} k_{j}$.
If $\sum_{j \in J^c} k_{j} > K$ then the score of
$c$ is strictly greater than the final score of $a$.
Therefore, the only way for $a$ to win is
if there exists a partition $\sum_{j \in J} k_{j} = K$
and  $\sum_{j \in J^c} k_{j} = K$.
In this case, $score(c) \leq KB+2n$,
$score(b) \leq KB + 2n$ and $score(a) \geq KB+2n$.
Hence, $a$ is a co-winner iff the \textsc{partition} instance
is a \textsc{Yes}-instance.

Suppose a partition $(J,J^c)$ of $A$ exists with $\sum_{j \in J} k_j = \sum_{j \in J^c} k_j$.
For the $j$th path, $j \in J$ we fix an order
$3j \prec 3j+1 \prec 3j+2$, where $ x\prec y$ means
$x$ votes before $y$.
For the $j$th path, $j \notin J$ we fix an order
$3j+1 \prec 3j \prec 3j+2$. This ensures that the weights of the $(3j+2)$th agents
in all paths are split equally between $b$ and $c$. Hence, $a$ is a co-winner.\myqed
\end{proof}

\begin{theorem}~\label{upw_bpg_alt_constant}
The unweighted possible winner problem is NP-complete even if
the number of candidates is constant,
the social network graph is bipartite, and each agent has
two preferred candidates.
\end{theorem}
\begin{proof}
We reduce from an instance of the {\sc 3-hitting set} problem.
For each element $q_j$, $j=0,\ldots,n-1$ we introduce 4
agents $4j,4j+1,4j+2$ and $4j+3$,
with  preferences $(c,b), (a,c), (b,c)$ and $(b,c)$, respectively.
Agents $4j, 4j+1, 4j+2$ and $4j+3$ form a path of friends.
We say that agents $4j,4j+1,4j+2$
and $4j+3$ represent the $j$th path that corresponds
to the $q_{j}$th element. In particular,
we refer to the $(4j+1)$th agent as an \emph{element-agent}, as
her decision corresponds to a selection of
the $q_{j}$th element into a hitting set.
For each set $S_i = (q_h,q_s,q_r)$, $i=1,\ldots,t$ we introduce $D$
agents $\{(4n-1) + D(i-1) + p\}$, $p=1,\ldots,D$, with  preferences $(b,a)$.
The $((4n-1) + D(i-1) +1)$th agent is a friend of the $(4h+1)$th, $(4s+1)$th and $(4r+1)$th agents.
Moreover,  $(4n-1) + D(i-1) + p$, $p=1,\ldots,D$ form a path of friends that starts
at $(4n-1) + D(i-1) + 1$ and ends at $(4n-1) + D(i-1) + D$.
We  refer to these as \emph{ set-agents}.
Finally, we introduce $B - k -Dt$ isolated agents
with  preferences $(a,c)$ and
$B - 2k$ isolated agents
with  preferences $(b,c)$,
where $B$ and $D > t$ are large integers such as $n^9$ and $n^4$.
We ask whether $a$ is a possible winner.
Figure~\ref{fig:proof_upw_bpg_alt_constant} illustrates the construction.
The basic score of $a$ is $B-k-Dt$
and of $b$ is $B-2k$.
The idea of the construction is that for $a$
to win it needs at least $Dt-k$ votes.
The construction ensures that at most $k$ of the $(4j+1)$th element-agents,
$j=0,\ldots,n-1$, can vote for $a$, otherwise $b$ beats $a$.
This corresponds to a selection of $k$ elements
in the hitting set. The
$Dt$ set-agents must all vote for $a$, otherwise $a$ loses,
which is possible iff a set of element-agents that selected $a$
corresponds to a hitting set.

\begin{figure}
\centering
\includegraphics[width=0.5\textwidth]{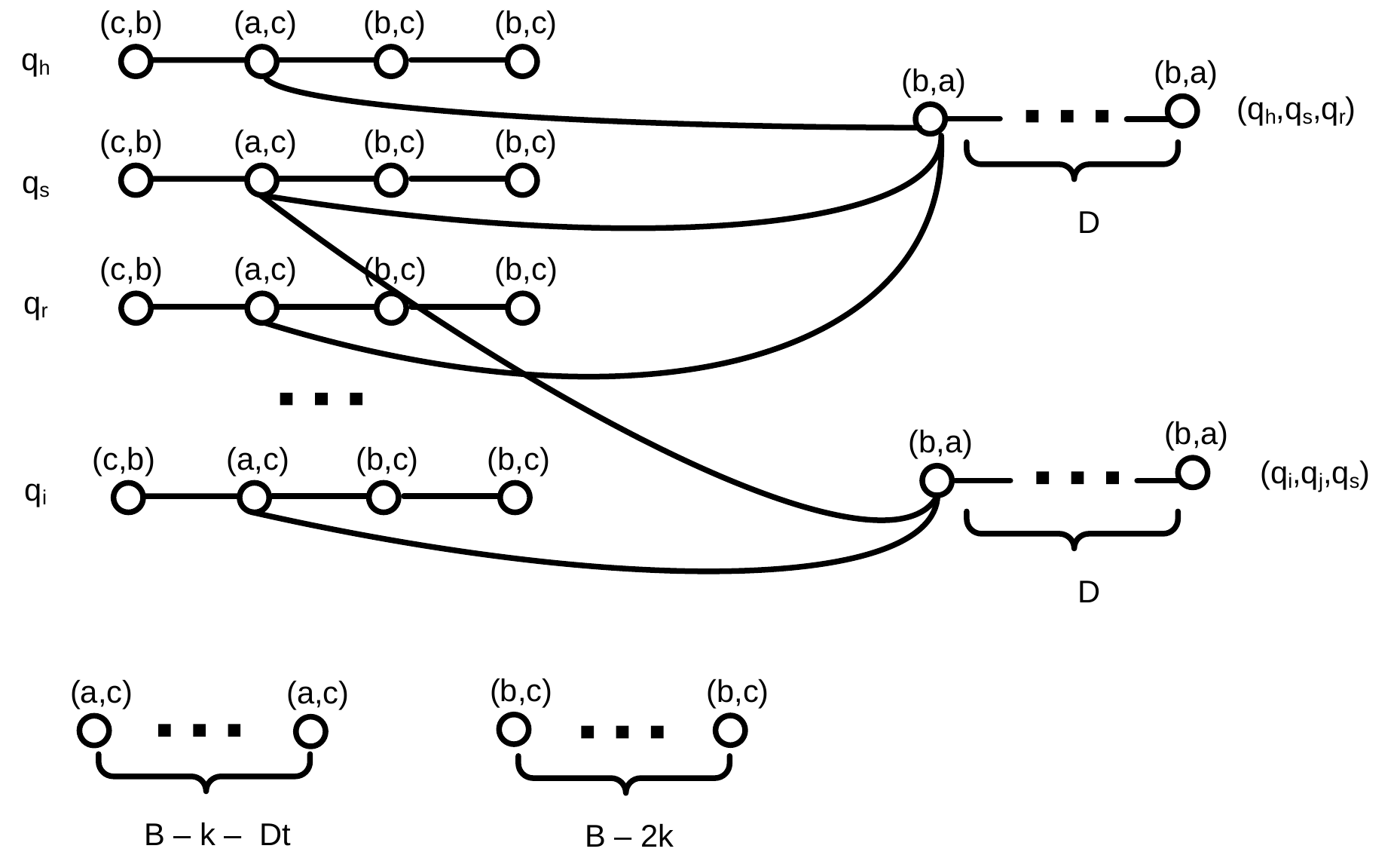}
\caption{\label{fig:proof_upw_bpg_alt_constant} The construction from Theorem~\ref{upw_bpg_alt_constant}}
\end{figure}

\medskip
\noindent\textbf{Select a set of elements.}
If the $(4j+1)$th element-agent in the $j$th path
selects the candidate $a$ then the agents
$(4j+2)$ and $(4j+3)$ will select their
 choice $b$.
Hence, increasing the score
of $a$ by 1 increases the score of $b$ by 2 if we only
consider voters in the $j$th path.
The basic score of  $a$ is $B-k-tD$,
the maximum number of points that $a$ can gain
from set-agents is $Dt$, and the basic
score of $b$ is $B-2k$; hence
at most $k$ element-agents can select $a$.

\medskip
\noindent\textbf{Check a hitting set.}
Suppose exactly $k'$ element-agents selected $a$
and the corresponding $k'$ elements cover $t'$ sets.
The remaining set of element-agents vote for $c$.
Hence, $Dt'$ set-agents
vote for $a$ and the remaining $(t-t')D$ vote for $b$.
Then the maximum score of $a$ is $B  - (k + Dt) +  (k'+t'D)$.
The maximum score of $b$ in this case is $B - 2k+2k' + (t-t')D$.
For $a$ to beat $b$ we need
$B  - (k + Dt) +  (k'+Dt') \geq B - 2k+2k'+ (t-t')D$ or
$2Dt' + k \geq k' + 2Dt$. As $D > t$,
this inequality holds iff $t' \geq t$.
Hence, $k'$ selected elements must form
a hitting set.
As at most $k$
element-agents are allowed to select $a$,
the problem has a solution iff there is a solution
to the hitting set problem.

\medskip
\noindent\textbf{Order construction.}
Let $H$ be a hitting set of size $k$. Then
$J = \{h : q_h \in H\}$ and
$ J^c = \{0,\dots,n-1\} \setminus J$.
First, the agents $\{4j,  \ldots, 4j+3\}$, $j \in J$
vote in the order $4j+1 \prec 4j \prec 4j+2 \prec 4j+3$,
so that each agent  selects his top choice.
Then all set-agents vote in the order
$(4n-1) +  1 \prec (4n-1) +  2 \prec  \ldots \prec (4n-1) + D(t-1) + D$.
As the set $J$
corresponds to the hitting set $H$, all set-agents vote for $a$.
Finally, the agents $\{4j,  \ldots, 4j+3\}$, $j \in J^c$,
vote in the order  $4j \prec 4j+1 \prec 4j+2 \prec 4j+3$,
so that each of these agents selects $c$.\myqed
\end{proof}

\begin{theorem}~\label{unw_bpg_alt_constant}
The unweighted necessary winner is co-NP-complete even if
the number of candidates is constant, the social network graph is
bipartite, and each agent has two preferred candidates.
 \end{theorem}
\begin{proof}
We use the construction from Theorem~\ref{upw_bpg_alt_constant}.
We ask if the candidate $b$ is a necessary winner.
This means $b$ does not lose to any other candidate under any order.
Note that $c$ cannot win the poll under any order
as the maximum possible score of $c$ is $4n$.
Hence, $b$ is a necessary winner iff there is no order
such that $a$ gets more points than $b$. From Theorem~\ref{upw_bpg_alt_constant}
if follows that $a$ gets more points than $b$ iff there exists a solution to the {\sc 3-hitting set} problem.\myqed
\end{proof}

\begin{theorem}~\label{upw_cng_alt_unbounded}
The unweighted possible winner problem is NP-complete
even if the social network graph is a disjoint union of
paths of length at most 1 and each agent has
two preferred candidates.
\end{theorem}
\begin{proof}

\begin{figure*}
\centering
\includegraphics[width=0.8\textwidth]{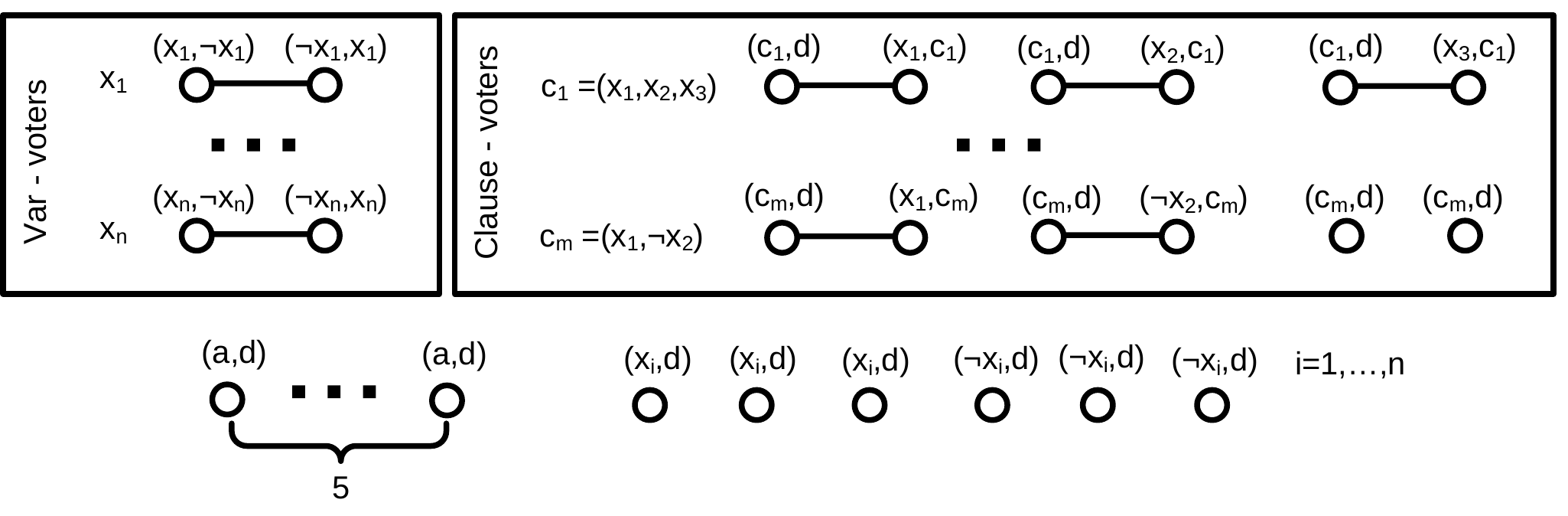}
\caption{\label{fig:proof_upw_cng_alt_unbounded} The construction from Theorem~\ref{upw_cng_alt_unbounded}}
\end{figure*}
We reduce from an instance of the {\sc $(3^{\leq},3^{\leq})$-SAT} problem.
We assume that the formula does not contain unit clauses and pure literals as
those can be removed during a preprocessing step. Therefore,
each variable occurs either twice positively and once negatively or
once positively and twice negatively. Hence, each variable can satisfy
at most 2 clauses.
For each literal, $x_i$ $(\bar{x}_i)$, $i=1,\ldots,n$,  we introduce a candidate labeled with
$x_i$ $(\bar{x}_i)$.
For each clause, $c_j$, $j=1,\ldots,m$, we introduce a candidate labeled with $c_j$.
Finally, we introduce a dummy candidate $d$ and the distinguished candidate $a$.
For each variable $x_i$, $i=1,\ldots,n$, we introduce
two \emph{var-agents}, $\{2i,2i+1\}$, with  preferences $(x_i,\bar{x_i})$ and $(\bar{x}_i,x_i)$,
respectively. Agents $2i$ and $2i+1$ are friends.
For each clause $c_j$, $j=1,\ldots,m$,
of length $3$, $c_j = (l_t,l_s,l_r)$, $j=1,\ldots,m$, $l_h \in \{x_h,\bar{x}_h\}$, $h \in \{t,s,r\}$,
we introduce 6 \emph{clause-agents}, $\{2n + 6j+1, \dots, 2n + 6j+6\}$, that we split into three groups of two agents,
one group for each literal in a clause. Agents in each group are friends. The first group
contains two agents with  preferences $(c_j, d)$ and $(l_t,c_j)$,
the second -- two agents with  preferences $(c_j, d)$ and $(l_s,c_j)$
and  the third -- two agents with  preferences $(c_j, d)$ and $(l_r,c_j)$.
For each clause $c_j$, $j=1,\ldots,m$
of length $2$, $c_j = (l_t,l_s)$, $j=1,\ldots,m$, $l_h \in \{x_h,\bar{x}_h\}$, $h \in \{t,s\}$,
we introduce 6 clause-agents: two groups of two agents for each literal in the clause as described above
and two isolated dummy agents with  preferences $(c_j,d)$.
Finally, we introduce $3$ isolated agents
with  preferences $(l_h,d)$, for each literal $l_h \in \{x_h,\bar{x_h}\}$, $h=1,\ldots,n$
and $5$ isolated agents
with  preferences $(a,d)$.
We ask whether $a$ is a possible winner.
Figure~\ref{fig:proof_upw_cng_alt_unbounded} illustrates the construction.

The basic score of $a$ is $5$, of a literal $l_h$, $l_h \in \{x_h,\bar{x_h}\}$, $h=1,\ldots,n$,
is $3$ and of a clause $c_j$ of size 2, $j \in \{1,\ldots,m\}$, is $2$.

\medskip
\noindent\textbf{Select an assignment.}
Consider a variable $x_i$
and the two corresponding
  var-agents, $2i$ and $2i+1$.
These agents make sure
that either $x_i$ or $\bar{x}_i$ gets two points exclusively.
As the basic score of $x_i$ and $\bar{x}_i$
is $3$, if $x_i$ ($\bar{x}_i$) gets 2 points
from var-agents then it is not allowed to get
any points from clause-agents.
We say that the candidate $x_i$ is selected by an assignment iff
$\bar{x}_i$ gets two points from var-agents
and $\bar{x}_i$ is selected otherwise.
We emphasize that candidates
that are not selected by an assignment are not allowed
to obtain any additional points from clause-agents.

\medskip
\noindent\textbf{Check an assignment.}
Consider a clause $c_j = (x_t,\bar{x}_s,x_r)$.
Due to clause-agents,
the candidate $c_j$ gets at least three points
from the corresponding clause-agents regardless
of the voting order. Moreover, the candidate $c_j$ can get
at most five points from these  clause-agents, otherwise $a$ loses.
Hence, at least one point
has to be given to one of the candidates $\{x_t,\bar{x}_s,x_r\}$.
Hence, at least one of these candidates must be selected to
the assignment. In other words, the corresponding literal
satisfies the clause $c_j$.
The analysis for clauses with two literals is similar.
Note that a candidate in an assignment can gain
at most two points from clause-agents.
In other words, it can satisfy at most two clauses,
which is the maximum number of clauses that a variable
can satisfy in the {\sc $(3^{\leq},3^{\leq})$-SAT} problem
that we consider in the reduction.
Hence, $a$ wins iff there exists a solution of the {\sc $(3^{\leq},3^{\leq})$-SAT} problem.

\medskip
\noindent\textbf{Order construction.}
Let $L$ be the  literals in a satisfying assignment.
For $i=1,\ldots,n$, if $x_i \in L$ then the agent
$2i+1$ votes at position $i$ and, otherwise,
the agent $2i$ votes at position $i$.
This fixes the voting order of $n$ first agents.
Then all clause-agents cast their votes. Note that
as $L$ is a satisfying assignment, none of the
candidates $c_j$, $j=1,\ldots,m$ has more than
$5$ points. The voting order of the remaining agents
is arbitrary.\myqed
\end{proof}


\section{Conclusions}

We have introduced a general model of
social polls in which an agent's vote
is influenced by their friends in their
social graph that have already voted. 
We consider a particular instance of this
model in which influence is very simple:
an agent votes for their most preferred 
candidate unless one of their $k$ most
preferred candidates has already received
a majority of votes from their friends who 
have already voted. We consider how to
compute who can possibly or necessarily 
win such a social poll depending on the
order of the agents yet to vote. These
problems are closely related to a number
of questions regarding control and manipulation
of such votes. Our results show that
the computational complexity of the possible and necessary winner
problems depend on the structure of the underlying social graph
and the number of candidates.
The possible winner problem
is NP-hard to compute in general, even under
strong restrictions on the structure of the social graph.
By comparison, the necessary winner
problem is often computationally easier
to compute. For instance, it
is polynomial to compute if the social
graph has bounded treewidth.

\paragraph{Acknowledgments}
NICTA is funded by
the Australian Government as represented by
the Department of Broadband, Communications and the Digital Economy and
the Australian Research Council.
Serge Gaspers acknowledges support from the Australian Research Council (grant DE120101761).

{\small
\bibliographystyle{plain}
\bibliography{literature}

\begin{thebibliography}{10}

\bibitem{alon12}
N.~Alon, M.~Babaioff, R.~Karidi, R.~Lavi, and M.~Tennenholtz.
\newblock Sequential voting with externalities: herding in social networks.
\newblock In {\em Proceedings of the 13th ACM Conference on Electronic
  Commerce}, EC '12, pages 36--36. ACM, 2012.

\bibitem{Battaglini07}
M.~Battaglini, R.~Morton, and T.~Palfrey.
\newblock {Efficiency, Equity, and Timing of Voting Mechanisms}.
\newblock {\em American Political Science Review}, 101(03):409--424, 2007.

\bibitem{vanBevernFGR12}
Ren{\'e}~van Bevern, Michael~R. Fellows, Serge Gaspers, and Frances~A.
  Rosamond.
\newblock How applying {M}yhill-{N}erode methods to hypergraphs helps mastering
  the art of trellis decoding.
\newblock Technical Report CoRR abs/1211.1299, arXiv, 2012.

\bibitem{BikhchandaniHW92}
Sushil Bikhchandani, David Hirshleifer, and Ivo Welch.
\newblock A theory of fads, fashion, custom, and cultural change as
  informational cascades.
\newblock {\em Journal of Political Economy}, 100(5):992--1026, 1992.

\bibitem{Bodlaender96}
Hans~L. Bodlaender.
\newblock A linear-time algorithm for finding tree-decompositions of small
  treewidth.
\newblock {\em SIAM Journal on Computing}, 25(6):1305--1317, 1996.

\bibitem{BodlaenderFW92}
Hans~L. Bodlaender, Michael~R. Fellows, and Tandy Warnow.
\newblock Two strikes against perfect phylogeny.
\newblock In {\em Proceedings of the 19th International Colloquium on Automata,
  Languages and Programming (ICALP 1992)}, volume 623 of {\em LNCS}, pages
  273--283. Springer, 1992.

\bibitem{csljac2m007}
V.~Conitzer, T.~Sandholm, and J.~Lang.
\newblock When are elections with few candidates hard to manipulate.
\newblock {\em Journal of the Association for Computing Machinery}, 54, 2007.

\bibitem{cwxaaai11}
V.~Conitzer, T.~Walsh, and L.~Xia.
\newblock Dominating manipulations in voting wih partial information.
\newblock In W.~Burgard and D.~Roth, editors, {\em Proceedings of the
  Twenty-Fifth AAAI Conference on Artificial Intelligence (AAAI 2011)}. AAAI
  Press, 2011.

\bibitem{Dekel00}
E.~Dekel and M.~Piccione.
\newblock Sequential voting procedures in symmetric binary elections.
\newblock {\em Journal of Political Economy}, 108(1):pp. 34--55, 2000.

\bibitem{Diestel10}
Reinhard Diestel.
\newblock {\em Graph Theory}, volume 173 of {\em Graduate Texts in
  Mathematics}.
\newblock Springer Verlag, New York, 4th edition, 2010.

\bibitem{DoerrFF12}
Benjamin Doerr, Mahmoud Fouz, and Tobias Friedrich.
\newblock Why rumors spread so quickly in social networks.
\newblock {\em Communications of the ACM}, 55(6):70--75, 2012.

\bibitem{DowneyF99}
Rodney~G. Downey and Michael~R. Fellows.
\newblock {\em Parameterized Complexity}.
\newblock Monographs in Computer Science. Springer, New York, 1999.

\bibitem{fhhcacm10}
P.~Faliszewski, E.~Hemaspaandra, and L.A. Hemaspaandra.
\newblock Using complexity to protect elections.
\newblock {\em Communications of the ACM}, 53(11):74--82, 2010.

\bibitem{fpaimag10}
P.~Faliszewski and A.D. Procaccia.
\newblock {AI}'s war on manipulation: Are we winning?
\newblock {\em AI Magazine}, 31(4):53--64, 2010.

\bibitem{FellowsGR11}
Michael~R. Fellows, Serge Gaspers, and Frances Rosamond.
\newblock Multivariate complexity theory.
\newblock In Edward~K. Blum and Alfred~V. Aho, editors, {\em Computer Science:
  The Hardware, Software and Heart of It}, chapter~13, pages 269--293.
  Springer, 2011.

\bibitem{FellowsL89}
Michael~R. Fellows and Michael~A. Langston.
\newblock An analogue of the {M}yhill-{N}erode theorem and its use in computing
  finite-basis characterizations.
\newblock In {\em Proceedings of the 30th Annual Symposium on Foundations of
  Computer Science (FOCS 1989)}, pages 520--525, 1989.

\bibitem{FlumG06}
J{\"o}rg Flum and Martin Grohe.
\newblock {\em Parameterized Complexity Theory}, volume XIV of {\em Texts in
  Theoretical Computer Science. An EATCS Series}.
\newblock Springer, Berlin, 2006.

\bibitem{GareyJ79}
Michael~R. Garey and David~R. Johnson.
\newblock {\em Computers and Intractability}.
\newblock W. H. Freeman and Company, San Francisco, 1979.

\bibitem{gs1}
A.~Gibbard.
\newblock Manipulation of voting schemes: A general result.
\newblock {\em Econometrica}, 41:587--601, 1973.

\bibitem{HararyP73}
Frank Harary and Edgar~M. Palmer.
\newblock {\em Graphical Enumeration}.
\newblock Academic Press, 1973.

\bibitem{Kloks94}
Ton Kloks.
\newblock {\em Treewidth, Computations and Approximations}, volume 842 of {\em
  Lecture Notes in Computer Science}.
\newblock Springer, 1994.

\bibitem{klijcai2005}
K.~Konczak and J.~Lang.
\newblock Voting procedures with incomplete preferences.
\newblock In {\em Proceedings of the IJCAI-2005 workshop on Advances in
  Preference Handling}, 2005.

\bibitem{Liskovets08}
Valery~A. Liskovets.
\newblock More on counting acyclic digraphs.
\newblock Technical Report 0804.2496 [math.CO], arXiv, 2008.

\bibitem{MaudetPVR12}
Nicolas Maudet, Maria~Silvia Pini, Kristen~Brent Venable, and Francesca Rossi.
\newblock Influence and aggregation of preferences over combinatorial domains.
\newblock In {\em Proceedings of the 11th International Conference on
  Autonomous Agents and Multiagent Systems (AAMAS 2012)}, pages 1313--1314.
  IFAAMAS, 2012.

\bibitem{Niedermeier06}
Rolf Niedermeier.
\newblock {\em Invitation to Fixed-Parameter Algorithms}.
\newblock Oxford Lecture Series in Mathematics and its Applications. Oxford
  University Press, Oxford, 2006.

\bibitem{oeaamas12}
S.~Obraztsova and E.~Elkind.
\newblock Optimal manipulation of voting rules.
\newblock In {\em 11th International Conference on Autonomous Agents and
  Multiagent Systems (AAMAS 2012)}, 2012.

\bibitem{prvwijcai2007}
M.~Pini, F.~Rossi, B.~Venable, and T.~Walsh.
\newblock Incompleteness and incomparability in preference aggregation.
\newblock In Manuela~M. Veloso, editor, {\em Proceedings of the 20th
  International Joint Conference on Artificial Intelligence (IJCAI-2007)},
  pages 1464--1469, 2007.

\bibitem{RobertsonS84}
Neil Robertson and Paul~D. Seymour.
\newblock Graph minors {III}: Planar tree-width.
\newblock {\em Journal of Combinatorial Theory, Series B}, 36(1):49--64, 1984.

\bibitem{Robinson73}
R.~W. Robinson.
\newblock Counting labeled acyclic digraphs.
\newblock In {\em New Directions in the Theory of Graphs}, pages 239--273.
  Academic Press, 1973.

\bibitem{gs2}
M.~Satterthwaite.
\newblock Strategy-proofness and {Arrow's} conditions: Existence and
  correspondence theorems for voting procedures and social welfare functions.
\newblock {\em Journal of Economic Theory}, 10:187--216, 1975.

\bibitem{Tovey84}
Craig~A. Tovey.
\newblock A simplified np-complete satisfiability problem.
\newblock {\em Discrete Applied Mathematics}, 8(1):85 -- 89, 1984.

\bibitem{waaai2007}
T.~Walsh.
\newblock Uncertainty in preference elicitation and aggregation.
\newblock In {\em Proceedings of the 22nd National Conference on AI}, pages
  3--8, 2007.

\bibitem{xcaaai08}
L.~Xia and V.~Conitzer.
\newblock Determining possible and necessary winners under common voting rules
  given partial orders.
\newblock In D.~Fox and C.P. Gomes, editors, {\em Proceedings of the
  Twenty-Third AAAI Conference on Artificial Intelligence (AAAI 2008)}, pages
  196--201. AAAI Press, 2008.

\end{thebibliography}
}

\end{document}